\documentclass[10pt, doublecolumn]{IEEEtran}
\usepackage{epsfig,latexsym}
\usepackage{float}
\usepackage{indentfirst}
\usepackage{amsmath}
\usepackage{bm}
\usepackage{amssymb}
\usepackage{times}

\usepackage{algorithm}
\usepackage[noend]{algpseudocode}
\usepackage{subfigure}
\usepackage{psfrag}
\usepackage{hyperref}
\usepackage{cite}
\usepackage{lastpage}
\usepackage{fancyhdr}
\usepackage{color}
 \usepackage{amsthm}
\usepackage{bigints}
\usepackage{booktabs}
\usepackage{multirow}
\usepackage{siunitx}

\newtheorem{Lemma}{Lemma}
\newtheorem{Corollary}{Corollary}
\newtheorem{lemma}[Lemma]{$\mathbf{Lemma}$}

\newtheorem{corollary}[Corollary]{$\mathbf{Corollary}$}

\newcounter{problem}
\newcounter{save@equation}
\newcounter{save@problem}
\makeatletter

\begin{document}
\title{\vspace{-1em} \LARGE{ NOMA-Assisted Grant-Free Transmission: \\How to Design   Pre-Configured SNR Levels?  }}

\author{ Zhiguo Ding,   Robert Schober,    Bayan Sharif,  and H. Vincent Poor   \thanks{ 
  
\vspace{-2em}

 Z. Ding and Bayan Sharif are with Department of Electrical Engineering and Computer Science, Khalifa University, Abu Dhabi, UAE.  Z. Ding is also with Department of Electrical and Electronic Engineering, University of Manchester, Manchester, UK.   R. Schober is with the Institute for Digital Communications,
Friedrich-Alexander-University Erlangen-Nurnberg (FAU), Germany. H. V. Poor is  with the Department of
Electrical and Computer Engineering, Princeton University, Princeton, NJ 08544,
USA.

  }\vspace{-2.6em}}
 \maketitle

\vspace{-1em}
\begin{abstract}
An effective way to realize non-orthogonal multiple access (NOMA) assisted grant-free transmission is  to first create multiple receive signal-to-noise ratio (SNR) levels and then serve multiple grant-free  users by employing  these   SNR levels as bandwidth resources. These SNR levels need to be pre-configured prior to the grant-free transmission and have great impact on the performance of grant-free networks. The aim of this letter is to illustrate different designs for configuring the SNR levels and  investigate their impact    on  the performance of  grant-free transmission, where   age-of-information   is used as the performance metric. The presented analytical and simulation results demonstrate the   performance gain achieved by NOMA over  orthogonal multiple access, and also reveal the relative merits  of  the considered    designs for   pre-configured SNR levels.  
\end{abstract}\vspace{-0.5em}

\begin{IEEEkeywords}
Grant-free transmission,   non-orthogonal multiple access (NOMA), age of information (AoI). 
\end{IEEEkeywords}
\vspace{-1em} 

\section{Introduction}
Grant-free transmission is a crucial feature of the sixth-generation (6G) network   to support   various important services, including  ultra-massive machine type communications (umMTC) and enhanced ultra reliable and low latency communications   (euRLLC) \cite{you6g,9031550}. Non-orthogonal multiple access (NOMA) has been recognized as a promising enabling technique to support grant-free transmission, and there are different realizations of NOMA-assisted grant-free transmission. For example, cognitive-radio inspired NOMA can be used to ensure that the bandwidth resources, which normally would be solely occupied by grant-based users, are used to admit   grant-free users \cite{8662677,9986037,9837307}. Power-domain NOMA has also  been shown   effective to realize grant-free transmission, where successive interference cancellation (SIC)   is carried out    dynamically according to the users' channel conditions
\cite{8533378,9707731,9782529}. 

The aim of this letter is to focus on the application of NOMA assisted random access for grant-free transmission \cite{jsacnoma10}.  Unlike the other aforementioned forms of NOMA, NOMA assisted random access ensures that grant-free transmission can be supported without requiring  the base station to have access to the users' channel state information (CSI).  The key idea of  NOMA assisted random access is to first create multiple receive signal-to-noise ratio (SNR) levels and then serve users by employing  these  SNR levels as bandwidth resources.  These SNR levels need to be pre-configured prior to the grant-free transmission and have great impact on the performance of grant-free networks. In the literature, there exist two    SNR-level designs, termed Designs I and II, respectively. Design I is based on a pessimistic approach and is to ensure that a user's signal can be still decoded successfully, even if   all the remaining users choose the SNR level which contributes the most interference, an unlikely scenario  in practice \cite{9590503,crnomaaoigrat}. Design II is based on an optimistic approach, and assumes that each SNR level is to be selected by at most one user   \cite{jsacnoma10}. The advantage of Design II over Design I is that the SNR levels of Design II can be chosen much smaller than those of Design I, and hence are more affordable to the users. The advantage of Design I over Design II is that  a collision at one SIC stage does not cause all the SIC stages to fail. The aim of this letter is to study the impact of the two SNR-level designs on grant-free transmission, where the age-of-information (AoI) is used as the performance metric \cite{9380899,9508961}. As the AoI achieved by Design I has been analyzed in \cite{crnomaaoigrat}, this letter focuses on Design II, where a closed-form expression for the AoI achieved by NOMA with Design II and its high SNR approximation are obtained. The presented analytical and simulation results reveal  the   performance gain achieved by NOMA over  orthogonal multiple access (OMA). Furthermore, compared to Design I,   Design II is shown to be beneficial for reducing  the AoI at low SNR, but suffers a performance loss  at high SNR.  
 \vspace{-1em}
 \section{System Model}
 Consider a multi-user uplink   network, where $M$ users, denoted by ${\rm U}_m$, communicate with the same base station in a grant-free manner. In particular, each user generates its update at the beginning of a time frame which consists of $N$ time slots having duration  $T$ seconds each.  The users compete for   channel access to deliver their updates to the base station, where the probability of a transmission attempt, denoted by $\mathbb{P}_{\rm TX}$, is assumed to be identical for all users.  

With OMA, only a single user can be served in each time slot, whereas the benefit of using NOMA is  that multiple users can be served simultaneously.  Similar to \cite{crnomaaoigrat}, NOMA assisted random access is adopted for the implementation of NOMA. In particular, the base station pre-configures $K$ receive  SNR levels, denoted by $P_k$, where each user randomly chooses one of the $K$ SNR levels for its transmission with equal probability $\mathbb{P}_K$. If ${\rm U}_1$ chooses  $P_k$, ${\rm U}_1$ needs to scale its transmit signal   by $\frac{P_k}{|h_1^{j,n}|^2}$, where $h_m^{j,n}$ denotes ${\rm U}_{m}$'s channel gain in the $j$-th slot of the $n$-th frame. If the chosen SNR level is not feasible, i.e., $\frac{P_k}{|h_1^{j,n}|^2}>P$, the user simply keeps silent, where $P$ denotes the user's transmit power budget.  Each user is assumed to have access to its own CSI only, and the users' channels  are  assumed to be independent and identically   complex Gaussian distributed with zero mean and unit variance. 
\vspace{-1em}
\subsection{Two Designs to Configure  the SNR Levels, $P_k$}
  Recall that the SNR levels are configured prior to    transmission, which means that the SNR levels cannot be related to the users' instantaneous channel conditions, but  are solely determined by the users' target data rates. In the literature,  there exist two   SNR-level designs, as explained   in the following. For   illustrative purposes, assume that $P_1\geq \cdots \geq P_K$, i.e.,   SIC is carried out by decoding the user using  $P_i$ before decoding the one using  $P_j$, $i<j$.  For the trivial case where $M<K$, only the   $M$ smallest   SNR levels are used. 

\subsubsection{Design I} In \cite{9590503,crnomaaoigrat}, the receive  SNR levels are configured as follows: 
\begin{align}\label{cr rate}
\log\left(1+\frac{P_k}{1+(M-1)P_{k+1}} \right)=R, \quad 1\leq k \leq K-1,
\end{align}
and $\log\left(1+P_K\right)=R$, 
i.e.,    $P_K=2^R-1$ and $P_k=\left(2^R-1\right)\left(1+(M-1)P_{k+1}\right)$, where the users are assumed to have the same target data rate, denoted by $R$. The rationale behind this design is to ensure successful SIC  in the worst case, where  one user chooses $P_k$ and all the remaining users choose the SNR level  which contributes the most interference, i.e., $P_{k+1}$. This is a pessimistic assumption  since not all the remaining users make a transmission attempt, and it is unlikely for all users to choose the same SNR level.

\subsubsection{Design II} Alternatively,  the  SNR levels can also be configured as follows \cite{jsacnoma10}\footnote{A more sophisticated design is to introduce an auxiliary parameter, $\eta$, and integrate the two designs shown in \eqref{cr rate} and \eqref{cr rate2} together, i.e.,  $\log\left(1+\frac{P_k}{1+\eta P_{k+1}} \right)=R$, where an interesting direction for future research is to optimize $\eta$ for AoI reduction. }: 
\begin{align}\label{cr rate2}
\log\left(1+\frac{P_k}{1+P_{k+1}} \right)=R, \quad 1\leq k \leq K-1,
\end{align}
and $\log\left(1+P_K\right)=R$, 
i.e.,    $P_K=2^R-1$ and $P_k=\left(2^R-1\right)\left(1+P_{k+1}\right)$. Design II has the drawback that one collision in the $i$-th SIC stage can cause a failure to the earlier stages, i.e., the $j$-th SIC stage, $j<i$, since  there is more than one interference source for $P_i$. However, Design II offers the benefit that its SNR levels are less demanding than those for Design I, as can be seen from Table \ref{table1}. Recall that a user has to remain  silent if its chosen SNR level is not feasible.  Because  the SNR levels of Design I are large,   these SNR levels cannot be fully used by the users,  and hence  the number of   supported users is smaller than that for Design II.

\begin{table}\vspace{-1em}
{\scriptsize  \caption{  Receive SNR Levels Required By the Two Designs. }
\begin{tabular}{lccc|ccc}
    \toprule
    \multirow{2}{*}{ } &
      \multicolumn{3}{c}{Design I} &
      \multicolumn{3}{c}{Design II}  \\
      & {$R=1$} &{$R=1$} &{$R=2$} & {$R=1$ } & {$R=1$} & {$R=2$}   \\
       & {$M=5$} &{$M=10$} &{$M=5$} & { $M=5$} & {$M=10$} & {$M=5$}   \\
      \midrule
    $P_1$ & 85 & 820 & 5655 & 8 & 8 & 192  \\
   $P_2$ & 21 & 91 & 471 &4 & 4 & 48   \\
    $P_3$ & 5 & 10 & 39 & 2 & 2 & 12  \\
    $P_4$ & 1 & 1 & 3 &    1 & 1 & 3 \\
    \bottomrule
  \end{tabular}\label{table1}}\vspace{-2em}
\end{table}
\vspace{-1em}
\subsection{AoI of Grant-Free  Transmission}
For grant-free transmission, the AoI is an important metric to measure how frequently a user can update the base station. In particular,  an effective grant-free transmission scheme needs to ensure that the collisions among the users can be effectively   reduced, and the base station can be frequently updated, which makes the AoI an ideal metric.  Without loss of generality, ${\rm U}_1$ is focused on as the tagged user, and its average AoI is defined as follows  \cite{9380899,9508961}:
\begin{align}
\bar{\Delta} = \underset{W\rightarrow \infty}{\lim} \frac{1}{W} \int^{W}_{0}\Delta(t)dt,
\end{align}
where $\Delta(t)$ denotes the time elapsed since the last successfully delivered update. As the AoI achieved by OMA and NOMA with Design I has been analyzed in \cite{crnomaaoigrat}, the AoI achieved by Design II will be focused on in this letter. 

\section{AoI Performance Analysis}

To facilitate the AoI analysis, denote     the time internal between the $(n-1)$-th and the $n$-th successful updates by $Y_{n}$, and denote   the time     for the $n$-th successful update to be  delivered to the base station by $S_{n}$, $n=1, 2, \cdots$.  By using the definition of the AoI,  $\bar{\Delta}$ can be   expressed as follows \cite{crnomaaoigrat}: $
\bar{\Delta} =
  \frac{\mathcal{E}\{S_{n-1}Y_n\}}{\mathcal{E}\{Y_n\}}+\frac{\mathcal{E}\{Y_n^2\}}{2\mathcal{E}\{Y_n\}}$, where    $\mathcal{E}\{\cdot\}$ denotes the expectation operation. With some algebraic manipulations, $\mathcal{E}\{S_{n-1}Y_n\}  
=\mathcal{E}\{S_{n}\}  \mathcal{E}\{Y_n\} -\mathcal{E}\{S_{n} ^2 \}+\mathcal{E}\{S_{n}\}  ^2$, $
\mathcal{E}\{Y_n\} =
  \frac{TN}{1- \mathbf{s}_0^T\mathbf{P}_M^N\mathbf{1}_{M\times 1} }
$ and   $
\mathcal{E}\{Y_n^2\} =
   N^2T^2 \frac{1+ \mathbf{s}_0^T\mathbf{P}_M^N\mathbf{1}_{M\times 1} }{\left(1- \mathbf{s}_0^T\mathbf{P}_M^N\mathbf{1}_{M\times 1} \right)^2}
    + 2\mathcal{E}\left\{S_{n}^2 \right\}    - 2  \mathcal{E}\left\{S_{n}\right\}^2$, $
\mathcal{E}\{S_n\} = T\sum^{N}_{l=1}l  \frac{\mathbf{s}_0^T\mathbf{P}_M^{l-1} \mathbf{p}}{1- \mathbf{s}_0^T\mathbf{P}_M^N\mathbf{1}_{M\times 1} }$,
$
\mathcal{E}\{S_n^2\} =T^2 \sum^{N}_{l=1}l^2  \frac{\mathbf{s}_0^T\mathbf{P}_M^{l-1} \mathbf{p}}{1- \mathbf{s}_0^T\mathbf{P}_M^N\mathbf{1}_{M\times 1} }
$, $\mathbf{s}_0=\begin{bmatrix}1 & \mathbf{0}_{1\times (M-1)} \end{bmatrix}^T$, $\mathbf{0}_{m\times n}$ denotes an all-zero $m\times n$ matrix, $\mathbf{1}_{m\times n}$ is an all-one $m\times n$  matrix,   $\mathbf{p}=\mathbf{1}_{M\times 1} -\mathbf{P}_M\mathbf{1}_{M\times 1} $, and $\mathbf{P}_M$ is an $M\times M$ matrix to be explained  later. 

Recall that the considered   access  competition  among the users can be modelled as a Markov chain with $M+1$ states, denoted by $s_j$, $0\leq j\leq M$.  In particular, state $s_j$, $0\leq j \leq M-1$, denotes that  $j$ users   succeed in updating  the base station, and the tagged user is not one of the successful user.  $s_M$ denotes that the tagged user succeeds in updating  the base station. The state transition probability, denoted by ${\rm P}_{j,j+i}$,   is the probability from $s_j$ to $s_{j+i}$, $i\geq 0$ and $j+i\leq M-1$.    $\mathbf{P}_M$ is  an all zero matrix      except its element in the $(j+1)$-th row and $(j+i+1)$-th column is ${\rm P}_{j,j+i}$.       The calculation of the state transition probability is directly determined by the transmission strategy.  The following lemma provides  $ {\rm P}_{j,j+i} $ achieved by NOMA with Design II. 
\begin{lemma}\label{lemma1}
The state transition probability, $ {\rm P}_{j,j+i} $, achieved by NOMA with Design II is given by
  \begin{align}
 {\rm P}_{j,j+i} =&\sum^{M-j}_{m=i+1} {M-j \choose m} \mathbb{P}_{\rm TX}^m \left(1- \mathbb{P}_{\rm TX}\right) ^{M-j-m}
 \\\nonumber
 &\times \frac{M-j-i}{M-j}m \cdots   (m-i+1)\mathbb{P}_K^{m} \gamma_i \gamma_{m,i}
 \\\nonumber
 &+{M-j \choose i} \mathbb{P}_{\rm TX}^i \left(1- \mathbb{P}_{\rm TX}\right) ^{M-j-i}
  \frac{M-j-i}{M-j}i! \mathbb{P}_K^{i}\gamma_i ,
 \end{align}
 for $0\leq j\leq M-2$ and $1\leq i \leq \min\{K,M-1-j\}$,
 and  
   \begin{align}
 {\rm P}_{j,j} = 1-\sum^{\min\{K,M-j\}}_{i=1}\bar{\rm P}_{j,j+i},
\end{align}
for $0\leq j\leq M-1$,
where 
\begin{align}
\bar{\rm P}_{j,j+i} =&\sum^{M-j}_{m=i+1} {M-j \choose m} \mathbb{P}_{\rm TX}^m \left(1- \mathbb{P}_{\rm TX}\right) ^{M-j-m}
 \\\nonumber
 &\times m \cdots   (m-i+1)\mathbb{P}_K^{m} \gamma_i\gamma_{m,i}
 \\\nonumber
 &+{M-j \choose i} \mathbb{P}_{\rm TX}^i \left(1- \mathbb{P}_{\rm TX}\right) ^{M-j-i}
   i!\mathbb{P}_K^{i} \gamma_i
 \end{align}
 $\gamma_i= \hspace{-2em}\underset{\substack{k_1+\cdots+k_K=i \\\max\{ k_1, \cdots,k_K\}=1}}{
 \sum}\hspace{-2em}  \left(1-\mathbb{P}_{e,1}\right)^{k_1} \cdots    \left(1-\mathbb{P}_{e,K}\right)^{k_K}$,   $\gamma_{m,i}= \hspace{-2em}\underset{\substack{k_1+\cdots+k_K=m-i }}{
 \sum}\frac{(m-i)!}{k_1! \cdots k_K!}  \mathbb{P}_{e,1} ^{k_1} \cdots    \mathbb{P}_{e,K} ^{k_K}$, and $ \mathbb{P}_{e,k}=1-e^{-\frac{P_k}{P}}$. 
\end{lemma}
\begin{proof}
See Appendix \ref{approof1}. 
\end{proof}
At high SNR, i.e., $P\rightarrow \infty$, $\mathbb{P}\left(\frac{P_k}{|h_1^{j,n}|^2}>P\right)\rightarrow 0$, i.e.,  all SNR levels become affordable to the users, and hence the expressions for the state transition probability can be simplified as shown in   the following corollary. 
\begin{corollary}\label{corollary1}
At high SNR,  $ {\rm P}_{j,j+i}$ can be approximated as follows:
  \begin{align} 
 {\rm P}_{j,j+i} \approx& 
 \frac{(M-j-1)!}{(M-j-i-1)!} \mathbb{P}_{\rm TX}^i \left(1- \mathbb{P}_{\rm TX}\right) ^{M-j-i}
   \mathbb{P}_K^{i}\bar{\gamma}_i ,
 \end{align}
 and 
 \begin{align}
\bar{\rm P}_{j,j+i} \approx& \frac{(M-j)!}{(M-j-i)!} \mathbb{P}_{\rm TX}^i \left(1- \mathbb{P}_{\rm TX}\right) ^{M-j-i}
   \mathbb{P}_K^{i} \bar{\gamma}_i,
 \end{align}
 where $\bar{\gamma}_i= \hspace{-2em}\underset{\substack{k_1+\cdots+k_K=i \\\max\{ k_1, \cdots,k_K\}=1}}{
 \sum}\hspace{-2em} 1$.
 \end{corollary}
 {\it Remark 1:} The benefit of using NOMA for   AoI reduction can be illustrated  based on Corollary \ref{corollary1}. For the special case of $i=1$,  the high SNR approximation of $ {\rm P}_{j,j+1}$ can be expressed  as follows:
  \begin{align} 
 {\rm P}_{j,j+1} \approx& 
(M-j-1)\mathbb{P}_{\rm TX}^i \left(1- \mathbb{P}_{\rm TX}\right) ^{M-j-1}
   \mathbb{P}_K \bar{\gamma}_i .
 \end{align}
 If $\mathbb{P}_K =\frac{1}{K}$, $ {\rm P}_{j,j+1} $ can be simplified  as follows:
   \begin{align} 
 {\rm P}_{j,j+1} \approx& 
(M-j-1)\mathbb{P}_{\rm TX}^i \left(1- \mathbb{P}_{\rm TX}\right) ^{M-j-1},
 \end{align}
 which is exactly the same as that of the OMA case shown in \cite{crnomaaoigrat}. However, for OMA, $ {\rm P}_{j,j+i}=0$, $i>1$, whereas for NOMA,  $ {\rm P}_{j,j+i}>0$, $1<i\leq K$, which means that  with NOMA     more users can be served, and hence the AoI of NOMA will be smaller than that of OMA. 
 
 \section{Simulation Results}
  
 In Fig. \ref{fig1}, the AoI of the considered grant-free schemes is shown as a function of the number of users, $M$. As can be seen from the figure,  the use of NOMA transmission can significantly reduce the AoI compared to the OMA  case, particularly when there is a large number of users. This ability  to support massive connectivity  is valuable for umMTC which is the key use case of 6G networks.  The figure also demonstrates the accuracy of the analytical results developed in Lemma \ref{lemma1}.  In addition, Fig. \ref{fig1a} shows that at low SNR, the use of Design II yields a significant performance gain over Design I, particularly for large $M$. However, at high SNR, the use of Design I is more beneficial, as demonstrated in Fig. \ref{fig1b}. An interesting observation from Fig. \ref{fig1b} is that   for the special case of  $M=5$ and $P=30$ dB, the use of $K=2$ SNR levels yields a better performance than   $K=4$. This is due to the fact that   the used  choice    $\mathbb{P}_{\rm TX}=\min\left\{\frac{K}{M},1\right\}$ is not optimal, as can be explained by using  Corollary \ref{corollary1}, which shows that $ {\rm P}_{j,j+i}$ is a function of $ \left(1- \mathbb{P}_{\rm TX}\right) ^{M-1}$ at high SNR. For the special case of $M=5$ and $K=4$,         $\mathbb{P}_{\rm TX}=\frac{4}{5}$,  and hence $ \left(1- \mathbb{P}_{\rm TX}\right) ^{M-1}$  can be very small, which causes  the AoI of $K=4$ to be larger than that of $K=2$. We note that for large $M$,  the performance gain of NOMA over OMA can be always  improved  by increasing $K$, i.e., using   more SNR levels, as shown in Fig. \ref{fig1}.

  \begin{figure}[t] \vspace{-0.1em}
\begin{center}
\subfigure[ $P=0$ dB]{\label{fig1a}\includegraphics[width=0.35\textwidth]{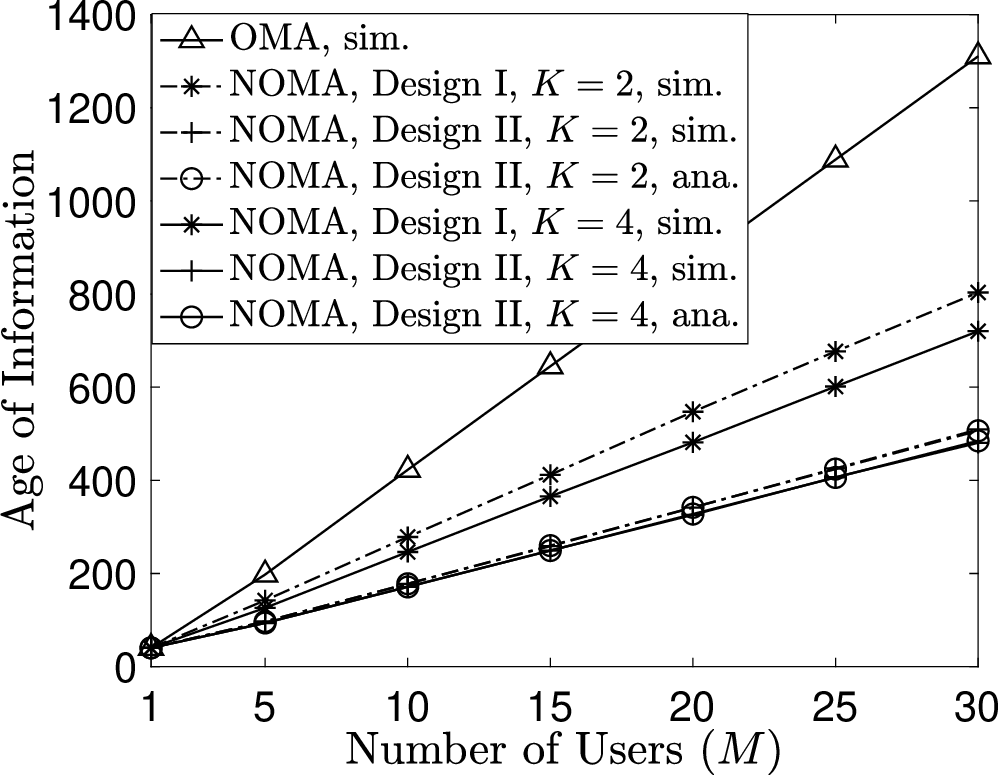}}\hspace{2em}
\subfigure[$P=30$ dB]{\label{fig1b}\includegraphics[width=0.35\textwidth]{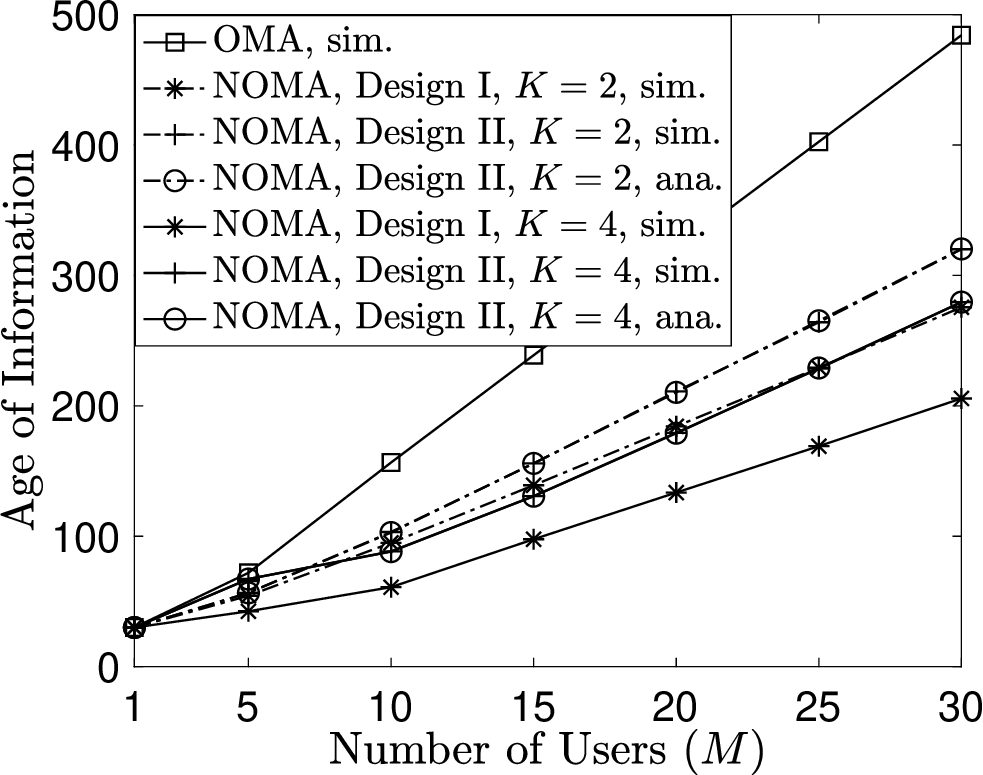}} \vspace{-1em}
\end{center}
\caption{ AoI achieved by the considered OMA and NOMA schemes as a function of $M$, where $T=6$, $R=1$, $N=8$, and $K=2$. For NOMA, $\mathbb{P}_{\rm TX}=\min\left\{\frac{K}{M},1\right\}$ and for OMA, $\mathbb{P}_{\rm TX}=\frac{1}{M-j}$, where $j$ is the number of users which have successfully delivered their updates to the base station.   \vspace{-0.1em} }\label{fig1}\vspace{-2em}
\end{figure}

  \begin{figure}[t] \vspace{-2em}
\begin{center}
\subfigure[ $K=2$  ]{\label{fig2a}\includegraphics[width=0.35\textwidth]{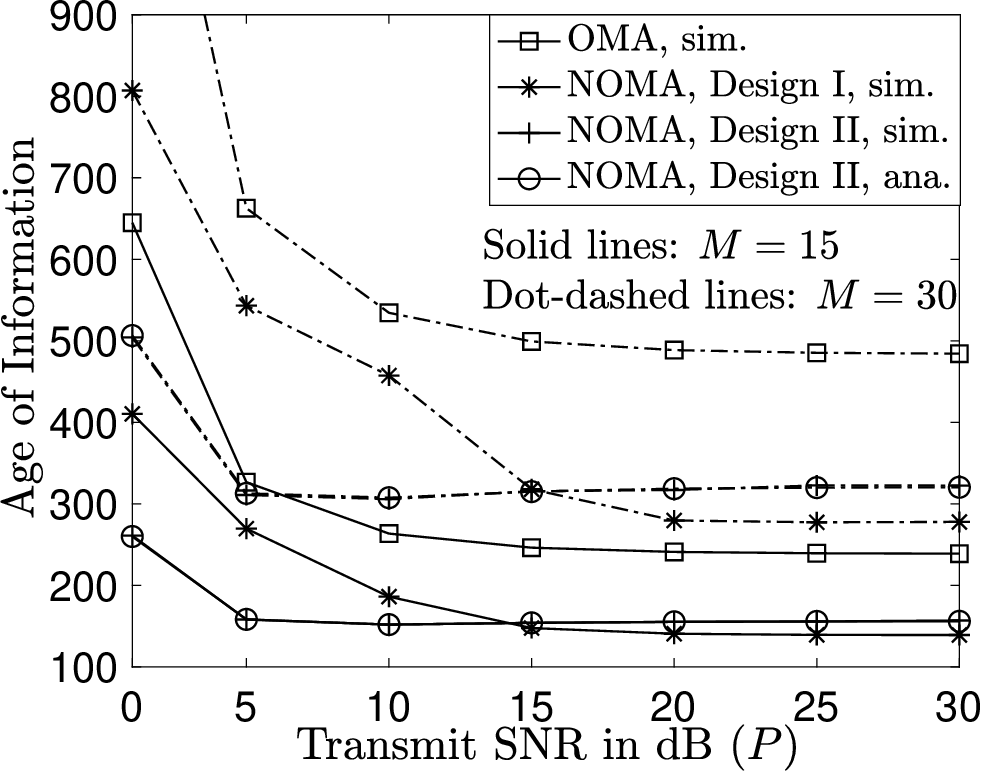}}\hspace{2em}
\subfigure[$K=4$  ]{\label{fig2b}\includegraphics[width=0.35\textwidth]{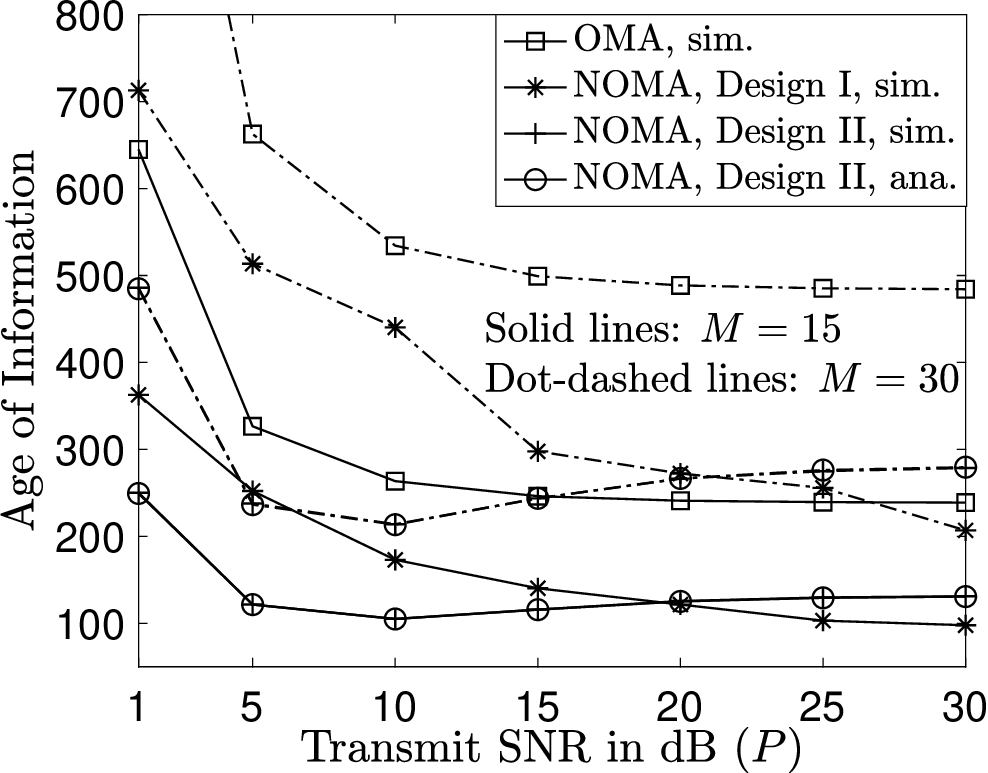}} \vspace{-1em}
\end{center}
\caption{ AoI achieved by the considered OMA and NOMA schemes as a function of SNR, $P$,  where $T=6$, $R=1$, $N=8$, and $K=2$. The same adaptive choices as in Fig. \ref{fig1} are used for $\mathbb{P_{\rm TX}}$.    \vspace{-0.1em} }\label{fig2}\vspace{-1.5em}
\end{figure}

 In order to better illustrate the impact of the transmit  SNR, $P$, on the performance of the two considered NOMA designs, Fig. \ref{fig2}   shows the AoI as a function of $P$. Fig. \ref{fig2} shows that regardless of the choices of the transmit SNR, the AoI of NOMA is always smaller than that of OMA, which is consistent with the observations from Fig. \ref{fig1}. In addition, Fig. \ref{fig2} also confirms the conclusion that Design II can outperform Design I at low SNR, but suffers a performance loss  at high SNR. The reason for Design II to outperform Design I at low SNR is that  the SNR levels required by Design I are more demanding than those of Design II, and hence may not be  affordable to the users at low SNR, i.e., $ \frac{P_k}{|h_1^{j,n}|^2}>P$. The reason for Design I to outperform Design II at high SNR is that,  at high SNR, all the levels of the two designs become affordable to the users, and   transmission failures are mainly caused by collisions, where unlike Design II, Design I ensures that a collision at the $i$-th SIC stage does not cause any failure to the $j$-th stage, $j<i$.  
 An interesting  observation from Fig. \ref{fig2} is that the AoI achieved by Design II
 may  even get  degraded by increasing SNR. This is because at low SNR, some users may find that their chosen SNR levels are not affordable, which reduces the number of active users and hence is helpful to reduce the AoI by avoiding   collisions.

As discussed previously, the transmission attempt probability, $\mathbb{P}_{\rm TX}$, is an important parameter for grant-free transmission.  In Fig. \ref{fig6}, we  show the AoI achieved by the considered schemes for different choices of  $\mathbb{P}_{\rm TX}$. In particular, we consider the adaptive choice,   $\mathbb{P}_{\rm TX}=\min\left\{\frac{K}{M},1\right\}$ for NOMA and   $\mathbb{P}_{\rm TX}=\frac{1}{M-j}$ for OMA, where $j$ is the number of users which have successfully delivered their updates to the base station. With the fixed choice, $\mathbb{P}_{\rm TX}$ is set as $0.05$.  As can be seen from the figure, with a given  choice of  $\mathbb{P}_{\rm TX}$, the AoI achieved by NOMA is worse than that of OMA for the special case of  low SNR and   small $M$. Nevertheless,   the performance gain of NOMA over OMA is still significant in general. In addition, the figure also shows that   the use of the adaptive choice of $\mathbb{P}_{\rm TX}$  yields a better performance than that of a fixed $\mathbb{P}_{\rm TX}$.  

\vspace{-1em}
\section{Conclusion}
In this letter, the application of NOMA-assisted random access to grant-free transmission has been studied, where   the two   SNR-level designs  and    their impact    on   grant-free networks have been investigated based on   the AoI. The presented analytical and simulation results show that the two NOMA designs outperform   OMA, and exhibit different behaviours  in the low and high SNR regimes.

  \begin{figure}[t] \vspace{-2em}
\begin{center}
\subfigure[ $P=0$ dB  ]{\label{fig6a}\includegraphics[width=0.35\textwidth]{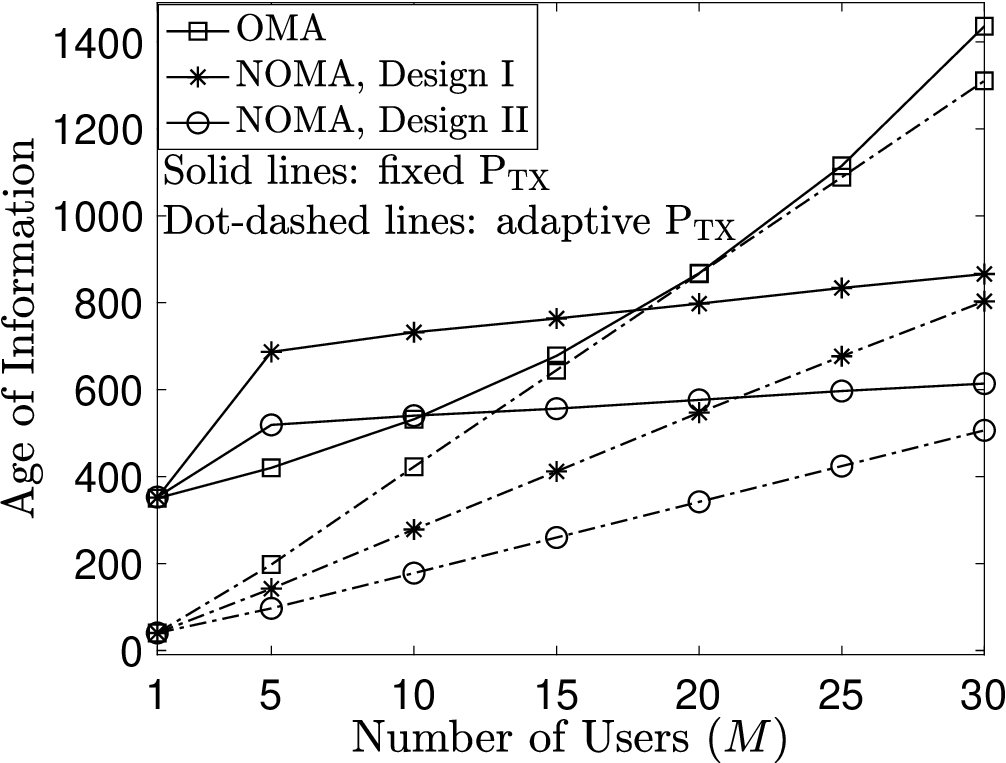}}\hspace{2em}
\subfigure[$P=30$ dB  ]{\label{fig6b}\includegraphics[width=0.35\textwidth]{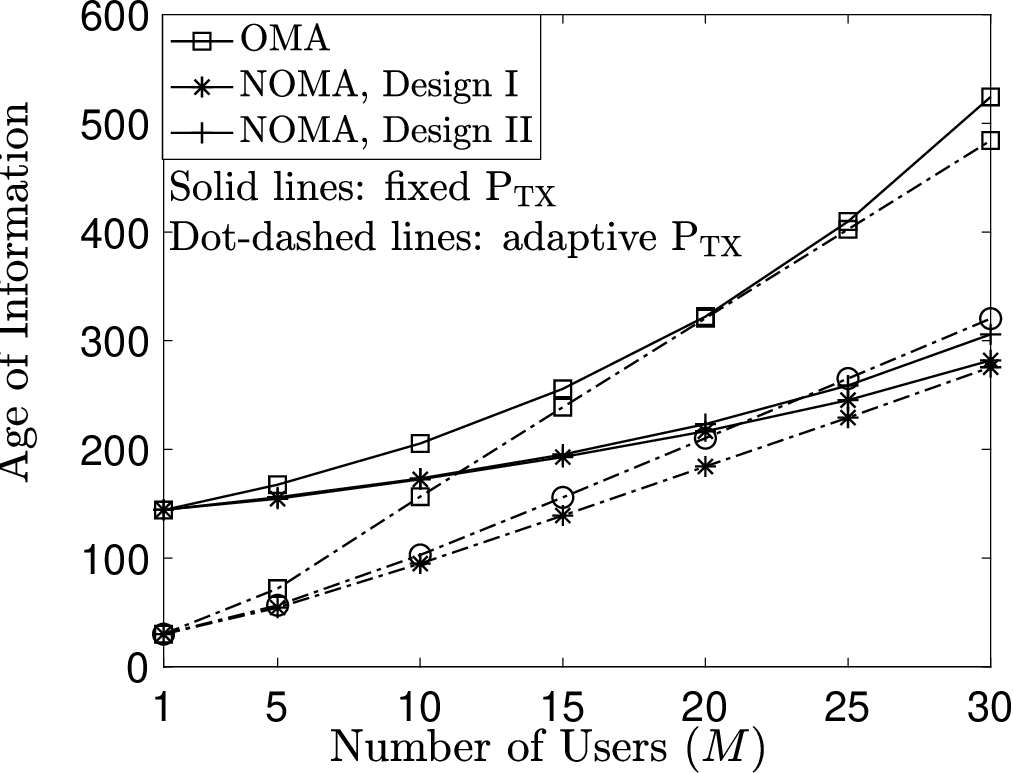}} \vspace{-1em}
\end{center}
\caption{  AoI achieved by the considered OMA and NOMA schemes with different choices of $\mathbb{P}_{\rm TX}$,  where $T=6$, $R=1$, $N=8$, and $K=2$.   \vspace{-0.1em} }\label{fig6}\vspace{-1.5em}
\end{figure}

\vspace{-1em} \appendices 
 \section{Proof for Lemma \ref{lemma1}}\label{approof1}
 
 
Suppose that among the $M$ users,  $j$ users have already successfully delivered their updates to the base station, but the tagged user, i.e., ${\rm U}_1$, still has not succeeded.    Define  $ {\rm E}_{j,j+i} $, $0\leq i \leq K$, as the event, that $i$ additional users succeed, but the tagged user, ${\rm U}_1$, is not among  the $i$ users. The key to studying  the AoI is to analyze the state transition  probabilities,  $ {\rm P}_{j,j+i} \triangleq \mathbb{P}( {\rm E}_{j,j+i} )$, $0\leq i \leq K$.
 
The expressions for  $ {\rm P}_{j,j+i} $, $0\leq i \leq K$, can be obtained from the following probabilities, $ \bar{\rm P}_{j,j+i} \triangleq \mathbb{P}( \bar{\rm E}_{j,j+i} ) $, $1\leq i \leq K$, where $ \bar{\rm E}_{j,j+i}$ denotes the event, in which among the $M-j$ users,  there are $i$ additional users which succeed in updating  their base station. Unlike for $  {\rm E}_{j,j+i} $, ${\rm U}_1$ can be one of the $i$ successful users for $\bar{\rm E}_{j,j+i} $.  

We note that not all the remaining $M-j$ users make a transmission attempt. By using the transmission attempt probability, $\mathbb{P}_{\rm TX}$,  $ \bar{\rm P}_{j,j+i}$ can be expressed as follows:
{\small \begin{align}\label{proof1}
 \bar{\rm P}_{j,j+i} =&\sum^{M-j}_{m=i+1} {M-j \choose m} \mathbb{P}_{\rm TX}^m \left(1- \mathbb{P}_{\rm TX}\right) ^{M-j-m} \bar{\rm P}_{m,i}, 
\end{align}}\hspace{-0.5em}
where $ \bar{\rm P}_{m,i}$ denotes the probability of the event, that among  $m$ active users, i.e., $m$ users making  a transmission attempt,  $i$   users succeed in updating their base station.
 
Without loss of generality, assume that  ${\rm U}_k$, $1\leq k \leq m$, are   the $m$ active users. In the following, we   focus on a particular event, denoted by $E_i$,  in which among  ${\rm U}_k$, $1\leq k \leq m$,   ${\rm U}_1$, $\ldots$, ${\rm U}_i$ are the $i$ successful users. Therefore, $\bar{\rm P}_{m,i}$ can be expressed as follows:
\begin{align}\label{proof2}
\bar{\rm P}_{m,i} = {m \choose i}\mathbb{P}(E_i),
\end{align}
where $ {m \choose i}$ is the number of events which have the same probability as     $E_i$.  

If Design I is used, the detection at the $n$-th SIC stage is affected by the $l$-th stage, $l<n$,  only, and   a collision which happens at a later stage, i.e., the $p$-th stage, $p>n$, has no impact. However, with Design II, a collision will cause all SIC stages to fail, which makes the performance analysis for Design II significantly different from the one shown in \cite{crnomaaoigrat}.

Considering  the difference between the two designs,  the fact that ${\rm U}_k$, $1\leq k \leq m$, are the $m$ active users, but only ${\rm U}_n$, $1\leq n \leq i$,   are successful has   the following two implications:
\begin{itemize}
\item There is no collision between ${\rm U}_n$, $1\leq n \leq i$, i.e., the $i$ successful users choose different SNR levels. In addition, each user finds its chosen SNR level feasible. 

\item Each of the failed users,    ${\rm U}_n$, $i+1\leq n \leq m$, finds out that its chosen SNR level is not feasible.  
\end{itemize}
The second implication  is the key to simplifying  the performance analysis,  and can be explained as follows. Without loss of generality, assume that ${\rm U}_{i+1}$ chooses $P_k$, and finds that $P_k$ is feasible. Because ${\rm U}_{i+1}$ is one of the active users, it will definitely make an attempt for transmission. Therefore, the only reason to cause this user's transmission to fail is a collision, i.e., another active user chooses the same SNR level as ${\rm U}_{i+1}$.  This collision at $P_k$ will cause a failure at the $k$-th SIC stage, as well as  the  following   SIC stages. More importantly, the collision at $P_k$ can also lead to a failure of the early SIC stages, due to the additional  interference caused by the  two simultaneous transmissions at  $P_k$. Define $E_{i,1}$ as the event where ${\rm U}_n$, $1\leq n \leq i$, successfully deliver their updates to the base station, and $E_{i,2}$ as the event where ${\rm U}_n$, $i+1\leq n \leq m$, fail to deliver their updates to the base station. The two aforementioned  implications   are also helpful in establishing  the independence between the two   events, $E_{i,1}$ and $E_{i,2}$, which leads to the following expression:
 \begin{align}\label{proof3}
\mathbb{P}(E_i) = \mathbb{P}(E_{i,1})  \mathbb{P}(E_{i,2}) .
\end{align}

In order to better illustrate how $\mathbb{P}(E_i)$ can be evaluated, define $\bar{E}_{i,1}$ as the particular event that ${\rm U}_n$ chooses $P_n$, $1\leq n \leq i$. By using the error probability defined in the lemma, $\mathbb{P}_{e,n}$, $\mathbb{P}(\bar{E}_{i,1})$ can be expressed as follows:
\begin{align}\label{pei1}
\mathbb{P}(\bar{E}_{i,1})= \mathbb{P}_K^{k_1} \left(1-\mathbb{P}_{e,1}\right)^{k_1} \cdots   \mathbb{P}_K^{k_K} \left(1-\mathbb{P}_{e,K}\right)^{k_K},
\end{align}
where $k_n=1$, for $1\leq n \leq i$, and $k_n=0$ for $i+1\leq n \leq K$. By using the general expression shown in \eqref{pei1} and enumerating all the possible choices of $k_n$, $1\leq n \leq K$,   $\mathbb{P}( {E}_{i,1}) $ can be evaluated as follows:
\begin{align}\nonumber
\mathbb{P}( {E}_{i,1}) &=i!\hspace{-2em} \underset{\substack{k_1+\cdots+k_K=i \\\max\{ k_1, \cdots,k_K\}=1}}{
 \sum} \hspace{-2em} \mathbb{P}_K^{k_1} \left(1-\mathbb{P}_{e,1}\right)^{k_1} \cdots   \mathbb{P}_K^{k_K} \left(1-\mathbb{P}_{e,K}\right)^{k_K}\\\label{pei1xx}&=i!\mathbb{P}_K^{i}
 \hspace{-2em} \underset{\substack{k_1+\cdots+k_K=i \\\max\{ k_1, \cdots,k_K\}=1}}{
 \sum} \hspace{-2em}  \left(1-\mathbb{P}_{e,1}\right)^{k_1} \cdots     \left(1-\mathbb{P}_{e,K}\right)^{k_K},
\end{align}
where $i!$ is the permutation factor since the event where  ${\rm U}_1$ and ${\rm U}_2$ choose $P_1$ and $P_2$, respectively,  is different from the event in which ${\rm U}_1$ and ${\rm U}_2$ choose $P_2$ and $P_1$, respectively. 

Similar to $ \mathbb{P}( {E}_{i,1}) $,  $\mathbb{P}( {E}_{i,2}) $ can be obtained as follows:
{\small \begin{align}\label{pei1xx2}
\mathbb{P}( {E}_{i,2}) &=\underset{\substack{k_1+\cdots+k_K\\=m-i }}{
 \sum}\frac{(m-i)!}{k_1! \cdots k_K!} \mathbb{P}_K^{k_1}  \mathbb{P}_{e,1} ^{k_1}\times\cdots\times  \mathbb{P}_K^{k_K}  \mathbb{P}_{e,K} ^{k_K},
\end{align}}
\hspace{-0.5em}where the multinomial coefficients $\frac{(m-i)!}{k_1! \cdots k_K!}$ is needed as explained in the following. Among the $m-i$ unsuccessful users, if $k_1$ users choose $P_1$, there are ${m-i \choose k_1}$ possible cases. For the remaining $m-i-k_1$ users, if $k_2$ users choose $P_2$, there are further ${m-i-k_1 \choose k_2}$ cases.  Therefore, the total number of cases for $k_n$ users to choose $P_n$, $1\leq n\leq K$, is given by ${m-i \choose k_1}\cdots {m-i-k_1-\cdots-k_K \choose k_K} = \frac{(m-i)!}{k_1! \cdots k_K!}$. It is interesting to point out that for ${E}_{i,1}$, the reason for having   coefficient $i!$ can   be explained in a similar manner, since $ \frac{i!}{k_1! \cdots k_K!}=i!$, if each $k_n$  is either one or zero.

By using \eqref{proof1}, \eqref{proof2}, and \eqref{proof3},   probability $ {\rm P}_{j,j}$
can be expressed as follows:
\begin{align}\nonumber
 {\rm P}_{j,j} =& 1-\sum^{K}_{i=1} \bar{\rm P}_{j,j+i}
 \\\nonumber
 =&1-\sum^{K}_{i=1} \sum^{M-j}_{m=i+1} {M-j \choose m} \mathbb{P}_{\rm TX}^m \left(1- \mathbb{P}_{\rm TX}\right) ^{M-j-m} \\\label{pjjj} &\times  {m \choose i} \mathbb{P}(E_{i,1})  \mathbb{P}(E_{i,2}).
\end{align} 
  By substituting   \eqref{pei1xx} and \eqref{pei1xx2} into \eqref{pjjj} and with some algebraic manipulations, the expression for $ {\rm P}_{j,j}$ can be explicitly obtained as shown in the lemma. 

By using the difference between $ {\rm E}_{j,j+i}$ and $ \bar{\rm E}_{j,j+i}$,   probability $ {\rm P}_{j,j+i}$ can be obtained  from $ \bar{\rm P}_{j,j+i}$ as follows:
\begin{align}\label{proof4}
  {\rm P}_{j,j+i} =&\sum^{M-j}_{m=i+1} {M-j \choose m} \mathbb{P}_{\rm TX}^m \left(1- \mathbb{P}_{\rm TX}\right) ^{M-j-m}\\\nonumber &\times \frac{M-j-i}{M-j}\bar{\rm P}_{m,i}. 
\end{align}
 The proof of the lemma is complete. 
 \vspace{-1em}
\bibliographystyle{IEEEtran}
\bibliography{IEEEfull,trasfer}
  \end{document}